\documentclass{llncs}
\usepackage{epsfig}
\usepackage[utf8]{inputenc}
\usepackage{tikz}
\usetikzlibrary{positioning}
\usepackage{microtype}%if unwanted, comment out or use option "draft"

\usepackage{amsmath}
\usepackage{amssymb}
\usepackage{appendix}
\usepackage{psfrag}
\usepackage{graphicx}
\usepackage{slashbox}
\usepackage{array}
\usepackage{etoolbox}
\usepackage{algorithm}
\usepackage{algorithmic}
\usepackage{xcolor}
\usepackage{epstopdf}

\usepackage{latexsym}
\usepackage{graphics}
\usepackage{epic}
\usepackage{eepic}
\usepackage{color}
\usepackage{theorem}
\newenvironment{appendix-theorem}[1]{\vspace{\theorempreskipamount}\noindent{\bf Theorem~#1~} \em }{\vspace{\theorempostskipamount}}
\newenvironment{appendix-lemma}[1]{\vspace{\theorempreskipamount}\noindent{\bf Lemma~#1~} \em }{\vspace{\theorempostskipamount}}

\newcommand{\even}{\mathcal{E}}
\newcommand{\un}{\mathcal{U}}
\newcommand{\EE}{\mathcal{E}}
\newcommand{\odd}{\mathcal{O}}
\newcommand{\U}{\mathcal{U}}
\newcommand{\unreach}{\mathcal{U}}
\newcommand{\A}{\mathcal{A}}
\newcommand{\Posts}{\mathcal{P}}
\newcommand{\posts}{\mathcal{P}}

\newcommand{\p}{\mathcal{P}} %added by Meghana.

\newcommand{\REM}[1]{}

\title{Dynamic Rank-Maximal Matchings}
\author{Prajakta Nimbhorkar$^1$
 \and
  Arvind Rameshwar V.$^2$\thanks{Part of the work was done when the author was a summer intern at Chennai Mathematical Institute.}
 }

\institute{
Chennai Mathematical Institute 
 {\tt (prajakta@cmi.ac.in)}
\and
Birla Institute of Technology and Science, Hyderabad Campus
{\tt (arvind.rameshwar@gmail.com)}
}

\begin{document}
\maketitle
\begin{abstract}
We consider the problem of matching applicants to posts where applicants have
preferences over posts. Thus the input to our problem is a bipartite graph 
$G=(\A\cup\p,E)$,
where $\A$ denotes a set of applicants, $\p$ is a set of posts, and there are 
ranks on edges which
denote the preferences of applicants over posts. A matching $M$ in $G$ is 
called {\em rank-maximal}
if it matches the maximum number of applicants to their rank~$1$ posts, 
subject to this
the maximum number of applicants to their rank~$2$ posts, and so on.

We consider this problem in a dynamic setting, where vertices and edges can be added 
and deleted at any point. Let $n$ and $m$ be the number of vertices and edges
in an instance $G$, and $r$ be the maximum rank used by any rank-maximal matching in $G$.
We give a simple $O(r(m+n))$-time 
algorithm
to update an existing rank-maximal matching under each of these changes. 
When $r=o(n)$, this is faster than recomputing a rank-maximal matching completely
using a known algorithm like that of Irving et al. \cite{IKMMP06}, which takes time 
$O(\min((r+n,r\sqrt{n})m)$.

\end{abstract}

\section{Introduction}
We consider matchings under one-sided preferences. The problem can be modeled as that of matching applicants to posts where applicants have 
preferences over posts. This problem has several important 
practical applications like allocation of graduates to training positions 
\cite{HZ79} and families to government housing \cite{Yuan96}.
The input to the problem consists of a bipartite graph $G = (\A \cup \p, E)$, 
where $\A$ is a set of applicants, $\p$ is a set of posts.
Each applicant has a subset of posts ranked 
in an order of preference. This is referred to as the {\em preference list} of the
applicant. An edge $(a,p)$ has rank $i$ if $p$ is an $i$th choice of $a$. An applicant can have
any number of posts at rank $i$, including zero.
Thus the edge-set $E$ can be partitioned as $E = E_1 \dot{\cup} \ldots \dot{\cup} E_r$, 
where 
$E_i$ contains the edges of rank~$i$. 
%$An edge $(a,p) \in E_i$ if $p$ is an 
%$i$th choice of $a$.
%An applicant $a$ prefers a post $p$ to $p'$ if, for some $i < j$, 
%$(a, p) \in E_i$ and $(a, p') \in E_j$. Applicant $a$ is indifferent between 
%$p$ and $p'$ if $i = j$.
%This ranking of posts by an applicant is called {\em the preference list} of the applicant.
%An applicant $a$ prefers post $p_i$ to post $p_j$ if the rank of post $p_i$ is
%smaller than the rank of post $p_j$ in $a$'s preference list.
%An applicant $a$ is indifferent between posts $p_i$ and $p_j$ if they have the same rank on $a$'s preference list.
%When applicants can be indifferent between posts, preference lists are said to 
%contain ties, else preference lists are strict.
%MN: 22/4 -- moved this to prelims.
%A matching $M$ of $G$ is a subset of edges, no two of which share an end-point.
%For a matched vertex $u$, we denote its partner in $M$ by $M(u)$.
%Given such an instance, our goal is 
%to compute a matching of applicants to posts that is {\em optimal} with respect to the preferences
%of the applicants. 

This problem 
has received lot of attention and there exist
several notions of optimality like pareto-optimality~\cite{ACMM04}, 
rank-maximality~\cite{IKMMP06}, popularity~\cite{AIKM07}, and fairness.
The notion of {\em rank-maximality} has been
 first studied by 
Irving~\cite{Irving03}, who called it {\em greedy matchings} and also gave an algorithm for
computing such matchings in case of strict lists. A rank-maximal matching 
matches maximum number of applicants to their rank~$1$ posts, subject to 
that, maximum number of applicants to their rank~$2$ posts and so on. 
Irving~et~al.\cite{IKMMP06} gave an $O(\min(n+r, r\sqrt{n})m)$-time algorithm 
to compute a rank-maximal matching.
Here $n = |\A| + |\p|$, $m = |E|$, and $r$ denotes the maximum rank on any edge in a rank-maximal
matching.
The weighted and capacitated versions of this problem have been
studied in \cite{KS06} and \cite{Paluch13} respectively.

%Given an instance $G = (\A \cup \p, E)$, a rank-maximal matching can be efficiently computed using the algorithm given
%by Irving et al. \cite{IKMMP06}. We remark that a given instance may admit more than one rank-maximal matching.

We consider the rank-maximal matching problem in a dynamic setting where
vertices and edges are added and deleted over time. The requirement of dynamic updates
in matchings has been well-studied in literature, with the motivation of updating 
an existing optimal matching without recomputing it completely. 
Dynamic updates are important in real-world applications as applicants matched to posts can leave their jobs, or new
applicants can apply for a job, or an applicant can acquire new skills and hence becomes eligible
for more posts. 
%Our algorithm crucially uses Irving et al.'s algorithm and the graphs it creates at each stage.
%In Irving et al.'s algorithm, at stage $i$, edges of rank $i$ are added to the instance and some edges which can not 
%belong to any rank-maximal matching are deleted. We show that addition or deletion of an applicant can lead to addition and deletion of several edges at each stage, however, at most one augmenting path is created at each stage. This helps us update each stage in time $O(m+n)$, thus total time taken is $O(r(m+n))$ where $r$ is the maximum rank on any edge in a rank-maximal matching.

{\bf Related work: } Bipartite matchings as well as popular matchings have been extensively studied in a dynamic setting \cite{OR10,BGS15,GP13,BHI15,BHN16} \cite{BHHKW15,AK06}. 
The algorithms for maintaining maximum matchings in dynamic
bipartite graphs maintain a matching under addition and deletion of edges that closely approximates the maximum cardinality matching,
and the update time is small i.e. sub-linear or even poly-logarithmic in the size of the graph. The 
algorithm of \cite{BHHKW15} maintains a matching that has an unpopularity factor of $(\Delta+k)$
with $O(\Delta+\Delta^2/k)$ amortized changes per round for addition or deletion of an edge, and 
$O(\Delta^2+\Delta^3/k)$ changes per round for addition and deletion of a vertex for any $k>0$. In contrast to
this, our algorithm maintains rank-maximal matchings exactly but needs $O(r(m+n))$ time for
each update. We describe our contribution below.
 
Recently, independent of our work, \cite{Ghosal17} give an $O(m)$ algorithm for updating rank-maximal matchings under addition and deletion of vertices using techniques similar to ours.

\subsection{Our Contribution}
We consider the problem of updating an existing rank-maximal matching when a vertex or edge is
added or deleted. 
We show the following in this paper:
\begin{theorem}\label{thm:main}
Given an instance of the rank-maximal matching problem with $n$ vertices and $m$ edges, there is an $O(r(m+n))$-time algorithm for updating a 
rank-maximal matching when a vertex or edge is added to or deleted from the instance. Here $r$
is the maximum rank used in any rank-maximal matching in the instance.
\end{theorem}
When $r=o(n)$, this is faster than recomputing a rank-maximal matching using the fastest known algorithm by Irving et al.\cite{IKMMP06}.

Our algorithm crucially uses Irving et al.'s algorithm and the graphs it creates for each stage.
In Irving et al.'s algorithm, at stage $i$, edges of rank $i$ are added to the instance and some edges which can not 
belong to any rank-maximal matching are deleted. We show that addition or deletion of a vertex or edge can lead to addition and deletion of several edges at each stage, however, at most one augmenting path is created at each stage. This helps us update each stage in time $O(m+n)$, thus total time taken is $O(r(m+n))$ where $r$ is the maximum rank on any edge in a rank-maximal matching. 

It is important to note that addition or deletion of even one edge can change an existing rank-maximal matching by as much as $\Omega(n)$ edges. We give an example in Appendix to show this. Also, addition or deletion of a vertex can potentially lead to addition or deletion of $\Omega(n)$ edges. In light of 
this, it is an interesting aspect of our algorithm that it avoids a complete recomputation of a rank-maximal matching. Also, in the instances that arise in practice, where there is a large number of applicants and posts, typically each applicant
ranks only a small subset of posts. Therefore our algorithm is useful for updating a rank-maximal matching in such instances substantially faster than recomputing it completely.

\subsection{Organization of the paper}
In Section \ref{sec:prelim}, we give some definitions and recall the algorithm of Irving et al.\cite{IKMMP06} for computing
a rank-maximal matching along with some of its properties. The preprocessing and an overview of the algorithm appear in Section \ref{sec:preproc}. The description and analysis of the algorithm is given
in Section \ref{sec:algo}. We discuss some related questions in Section \ref{sec:disc}.
%In \cite{GNN14}, a switching 
%graph characterization of rank-maximal matchings has been developed. This has turned out
%to be useful in several questions related to rank-maximal matchings. A natural question to ask
%is whether this characterization is also useful in dynamic setting. However, a switching graph
%is based on the {\em reduced graph} computed by Irving et al.'s algorithm, which is a subgraph
%of the input graph. Addition or deletion of a vertex can change this subgraph and hence the
%switching graph significantly. Therefore it is not clear whether the switching graph characterization
%can help in dynamic setting.
%Our algorithms require that the ranks on the remaining edges are not changed
%when a vertex and its incident edges are deleted or when a new vertex is 
%added along with ranks on its edges. Thus if a new post is added, an applicant
%may put it at a rank $i$, for some $i$, tied with his existing rank $i$ post.
%Similarly, if a post is deleted and if it was the only rank $i$ post for an
%applicant, the applicant will not have a rank $i$ post in the new instance.
%This is not a problem since neither our algorithms nor Irving et al.'s algorithm
%require preference lists to be contiguous.

\section{Preliminaries}\label{sec:prelim}
We recall some well-known definitions and terminology (see e.g. \cite{GNN14}).
A matching $M$ in a graph $G$ is a subset of edges, such that no two of them share a vertex.
For a matched vertex $u$, we denote by $M(u)$ its partner in $M$.
%We first review some well-known properties of maximum matchings in bipartite graphs. Then we define
%rank-maximal matchings, describe
%the algorithm of Irving~et~al.~\cite{IKMMP06} for computing a rank-maximal matching, and also recall some of its 
%invariants.

\paragraph{Properties of maximum matchings in bipartite graphs:} 
%MN: removed the paragraph to make more space.
%\paragraph{Properties of maximum matchings in bipartite graphs:} 
Let $G = (\A \cup \p, E)$ be a bipartite graph and let $M$ be a maximum matching in $G$.
The matching $M$ defines a partition of the vertex set $\A \cup \p$ into three 
disjoint sets, defined below:
\begin{definition}[Even, odd, unreachable vertices]\label{def:eou}
A vertex $v \in \A \cup \p$ is \emph {even} (resp. \emph {odd}) if there is an 
even (resp. odd) length alternating path with respect to $M$ from an unmatched 
vertex to $v$.
%has odd (resp. even) number of intermediate vertices.
A vertex $v$ is \emph {unreachable} if there is no alternating path from an unmatched vertex to $v$.
\end{definition}
%Denote by $\EE$, $\odd$, and $\U$ the sets of even, odd, and unreachable vertices, respectively, in $G$.
The following lemma is well-known in matching theory; see \cite{GGL95new} or \cite{IKMMP06} for a proof.

\begin{lemma}[\cite{GGL95new}]
\label{lem:node-class}
Let $\EE$, $\odd$, and $\U$ be the sets of even, odd, and unreachable vertices 
defined by a maximum matching $M$ in $G$. Then,
\begin {itemize}
\item [(a)] $\EE$, $\odd$, and $\U$ are disjoint, and are the same for all
 the maximum matchings in $G$.
\item [(b)] In any maximum matching of $G$, every vertex in $\odd$ is matched with a vertex in
$\EE$, and every vertex in $\U$ is matched with another vertex in $\U$.
The size of a maximum matching is $|\odd| + |\U|/2$.
\item  [(c)] No maximum matching of $G$ contains an edge with one end-point
 in $\odd$ and the other in $\odd \cup \U$.
Also, $G$ contains no edge with one end-point in $\EE$ and the other in $\EE \cup \U$.
\end {itemize}
\end{lemma}

\paragraph{Rank-maximal matchings:}
An instance of the rank-maximal matchings problem consists of a bipartite
graph $G=(\A\cup\p,E)$, where $\A$ is a set of applicants, $\p$ is
a set of posts, and applicants rank posts in order of their preference. That is
the input is a bipartite graph $G = (\A \cup \p, E)$ where the edges in 
$E$ can be
partitioned as $E_1 \cup E_2 \cup \ldots \cup E_r$. Here $E_i$ denotes
the edges of rank~$i$, and 
$r$ denotes the maximum rank any applicant assigns to a post. An edge $(a,p)$
has rank $i$ if $p$ is an $i$th choice of $a$. 

\begin{definition}[Signature]\label{def:sig}
The {\em
signature} of a matching $M$ is defined as an $r$-tuple $ \rho (M) = (x_1,\ldots,x_r)$
where, for each $1\leq i \leq r$, $x_i$ is the number of applicants who are 
matched to their $i$th rank post in $M$. 
\end{definition}
Let $M$, $M'$ be two matchings in $G$, with signatures 
$\rho(M) = (x_1, \ldots, x_r)$ 
and $\rho(M') = (y_1, \ldots, y_r)$.
%denote the signatures of $M$ and $M'$ respectively.
Define $M \succ M'$ if $x_i = y_i$ for $1 \le i < k \leq r$ and $x_k >  y_k$.

\begin{definition}[Rank-maximal matching]\label{def:rmm}
A matching $M$ in $G$ is  rank-maximal if 
$M$ has the maximum signature under the above ordering $\succ$.
\end{definition}
Observe that all the rank-maximal matchings in an instance 
have the same cardinality and the same signature.

\subsubsection{Construction of Rank-maximal Matchings:}
Now we recall Irving et al.'s algorithm \cite{IKMMP06} for computing a 
rank-maximal matching in a given instance $G=(\A\cup\p,E_1\cup\ldots\cup E_r)$.
The pseudocode of the same appears in Algorithm \ref{algo:Irving} and a description is given
below.
Recall that $E_i$ is the set of edges of rank $i$. 
% For the sake of convenience, we add 
%a last resort post $p_a$ for each applicant $a$ at rank $r+1$. Thus every applicant is 
%matched in a rank-maximal
%matching.

Let $G_i=(\A\cup\p, E_1\cup\ldots\cup E_i)$.
The algorithm involves $r$ stages, each stage $i$ considers
edges of rank at most $i$. The algorithm starts with $G_1'=G_1$. A maximum matching $M_1$ is computed in $G_1$. Then the vertices are labelled as even, odd, unreachable with respect to $M_1$. These sets of vertices
are called $\EE_1,\odd_1,\unreach_1$ respectively. The edges between
two vertices in $\odd_1$ or a vertex in $\odd_1$ and another in $\unreach_1$ can not belong to any maximum
matching, and hence they are deleted. Moreover, all the vertices in $\odd_1\cup\unreach_1$ have to be matched by any maximum matching in $G_1$. Therefore edges of rank more than $1$ incident on such vertices are also deleted from $G$. The resulting graph is called {\em the reduced graph} $G'_1$. The same process is repeated for each stage. Thus, at stage $i$, edges of rank $i$
are added to $G'_{i-1}$ i.e. the reduced graph of stage $i-1$ to get $G'_i$. 
A maximum matching $M_i$ is computed in $G'_i$ by augmenting the matching $M_{i-1}$ from the previous stage, and the vertices are partitioned into sets $\EE_i,\odd_i,\unreach_i$. Then the edges between two vertices in $\odd_i$ or between a vertex in $\odd_i$ and a vertex in $\unreach_i$ are deleted. Edges of rank more than $i$ incident on the vertices in $\odd_i\cup\unreach_i$ are also deleted from $G$. This is the reduced graph of stage $i$, called $G'_i$. It is shown in \cite{IKMMP06} that the matching $M_i$ is rank-maximal in $G_i$.

%MN:June-12-2014: Prajakta I have restored the algo as it is with scriptsize. See if you like it
%this way. Else the whole thing is commented below. 
%\vspace{-0.1in}
\begin{algorithm}[h]
\begin{algorithmic}[1]
\REQUIRE $G = (\A \cup \p, E_1 \cup E_2 \cup \dots \cup E_r)$.
\ENSURE A rank maximal matching $M$ in $G$.
\STATE Let $G_i = (\A \cup \p, E_1 \cup E_2 \cup \dots \cup E_i)$
\STATE Construct $G_1' = G_1$. Let $M_1$ be a maximum matching in $G_1'$.
\FOR {$i = 1 \ldots r $}
\STATE Partition $\A \cup \p$ as $\odd_i, \even_i, \un_i$ with respect to $M_i$ in $G_i'$.
\STATE \label{step:del1} Delete all edges of rank $j > i$ incident on vertices in $\odd_i \cup \un_i$.
\STATE \label{step:del2}Delete all edges from $G'_i$ between a node in $\odd_i$ and a node in $\odd_i \cup \un_i$.
\STATE Add edges in $E_{i+1}$ to $G_i'$; denote the resulting graph $G'_{i+1}$.
\STATE Compute a maximum matching $M_{i+1}$ in $G_{i+1}$ by augmenting $M_i$.
\ENDFOR
\STATE \label{step:del3}Delete all edges from $G'_{r+1}$ between a node in $\odd_{r+1}$ and a node in $\un_{r+1}$.
\STATE \label{step:graphGprime} Denote the graph $G'_{r+1}$ as $G'$.
\STATE Return a rank-maximal matching $M = M_{r+1}$.
\end{algorithmic}
\caption{An algorithm to compute a rank-maximal matching from~\cite{IKMMP06}.}
\label{algo:Irving}
\end{algorithm}
%\vspace{-0.1in}

%\noindent\fbox{
%\parbox{\textwidth}{For $i=1$ to $r$ do the following and output $M_{r+1}$:
%\begin{enumerate}
%\item Partition the vertices in $\A\cup\p$ into even, odd, and unreachable
%as in Definition \ref{def:eou} and call these sets $\EE_i,\odd_i,\U_i$ 
%respectively.
%\item Delete those edges in $E_j, j>i$, which are incident on nodes in $\odd_i
%\cup\U_i$. These are the nodes that are matched by every maximum matching in 
%$G'_i$.
%\item Delete all the edges from $G'_i$ between a node in 
%$\odd_i$ and 
%a node in $\odd_i\cup\U_i$. We refer to these edges as $\odd_i\odd_i$ and
%$\odd_i\U_i$ edges respectively. These are the edges which do not belong to any
%maximum matching in $G'_i$.
%\item Add the edges in $E_{i+1}$ to $G'_i$ and call
%the resulting graph $G'_{i+1}$.
%\item Determine a maximum matching $M_{i+1}$ in $G'_{i+1}$ by augmenting $M_i$.
%\end{enumerate}}
%}

%The algorithm constructs a graph $G'_r$. We construct a {\em reduced graph}
%$G'$ by deleting all the edges from $G'_r$
%between a node in $\odd_{r+1}$ and a node in $\odd_{r+1}\cup\U_{r+1}$.
%The graph $G'$ will be used in subsequent sections.

We note the following properties of Irving et al.'s algorithm:
\begin{enumerate}
\item [($I1$)] For every $1 \le i \le r$, every rank-maximal matching in $G_i$ 
is contained in $G'_i$.
\item[($I2$)] The matching $M_i$ is rank-maximal in $G_i$, and is a maximum 
matching in $G'_i$.
\item [($I3$)] If a rank-maximal matching in $G$ has signature $(s_1,\ldots,s_i,\ldots 
s_r)$ then $M_i$ has signature $(s_1,\ldots,s_i)$.
%\item [($I4$)] For every $i = 1, \ldots, r+1$, every rank-maximal matching matches all vertices in $\odd_i \cup \un_i$.
\item [($I4$)]The graphs $G'_i$, $1\leq i\leq r$
%MN June-15-2014. Removed the term reduced graph here to avoid confusion.
% referred to as the 
%reduced graph 
constructed at the end of iteration $i$ of Irving~et~ al.'s algorithm, 
and $G'$ are independent of 
the rank-maximal matching computed by the algorithm. This follows from 
Lemma \ref{lem:node-class} and invariant $I3$. 
\end{enumerate}

\section{Preprocessing and overview}\label{sec:preproc}
In the preprocessing stage, we store the information
necessary to perform an update in $O(r(m+n))$ time. The preprocessing
time is asymptotically same as that of computing a rank-maximal matching in a given
instance by Irving et al.'s algorithm viz. $O(\min((r+n,r\sqrt{n})m))$ and uses $O((m+n)\log n)$ storage.
\subsection{Preprocessing}
Given an instance of the rank-maximal matching problem, $G=(\A\cup\p,E)$ and ranks on edges, we execute
Irving et al.'s algorithm on $G$. (Algorithm~\ref{algo:Irving} from Section \ref{sec:prelim}.) 
Recall that $n$ is the number of vertices and $m$ is the number of edges.

We use the reduced graphs  $G'_i$ for $1\leq i\leq r$, where $G'_i=(\A\cup\Posts,E'_i)$, computed by
Algorithm \ref{algo:Irving} for updating a rank-maximal matching in $G$ on addition or deletion of an
edge or a vertex. 
If $M$ is a rank-maximal matching in $G$, then in each $G'_i$, we consider the matching $M_i=M\cap E'_i$.
By Invariant $(I2)$ from Section \ref{sec:prelim}, $M_i$ is rank-maximal in $G_i$.
When a vertex or an edge is added to or deleted from $G$, the goal is to emulate Algorithm~\ref{algo:Irving} on the new instance $H$ using the reduced graphs $G'_i$ for each $i$.

We prove in Lemma~\ref{lem:recon} below that we do not need to store the reduced graphs explicitly. The 
storage can be achieved by storing the original graph $G$ along with
some extra information for each stage. 
If a vertex becomes odd (respectively unreachable) at stage $i$ of
Algorithm~\ref{algo:Irving}, we store the number $i$ and one bit $0$ 
(respectively $1$) indicating that, at stage $i$, the vertex became odd
(respectively unreachable). 
For each edge, we store 
the stage at which it gets deleted, if at all. This takes $O((m+n)\log n)$ extra storage.
\begin{lemma}\label{lem:recon}
A reduced graph $G'_i$ of any stage $i$ of Algorithm~\ref{algo:Irving} can be completely
reconstructed from the stored information as described above. Moreover, this reconstruction
can be done in $O(m+n)$ time.
\end{lemma}
\begin{proof}
Edge-set $E'_i$ of $G'_i$ is a subset of $E_1\cup\ldots\cup E_i$. We go over all the edges in 
$E_1\cup\ldots\cup E_i$ and keep those edges in $E'_i$ which have not been deleted up to stage $i$.
This is precisely the information we have stored for each edge. As we go over each edge exactly
once, we need $O(m+n)$ time.
\end{proof}

\subsection{An overview of the algorithm}
Let $G$ be a given instance and let $H$ be the updated instance obtained by addition or deletion of an edge or a vertex. 
As stated earlier, the goal of our algorithm is to emulate Algorithm \ref{algo:Irving} on $H$ using
stored in the preprocessing step described above. Thus our algorithm constructs the reduced graphs
$H'_i$ for $H$ by updating the reduced graphs $G'_i$, and also a rank-maximal matching $M'$ in $H$ by updating a rank-maximal matching $M$ in $G$. We prove that the graphs $H'_i$ are same
as the reduced graphs that would be obtained by executing Algorithm \ref{algo:Irving} on $H$. 

The reduced graph $H'_i$ can be significantly different from the reduced graph $G'_i$ for a stage 
$i$. However, we show that there is
at most one augmenting path in $H'_i$ for any stage $i$. Thus each $H'_{i+1}$ and $M'_i$ can be obtained from $H'_{i}$, $G'_{i+1}$, and $M_i$ in linear time i.e. $O(m+n)$ time. 

We note that, in Irving et al.'s algorithm,
an applicant is allowed to have any number of posts of a rank $i$, including zero. Also, because of the one-sided preferences model, each edge has a unique rank associated with it. Thus addition of an
applicant is analogous to addition of a post. In both the cases, a new vertex is added to the instance,
along with the edges incident on it, and along with the ranks on these edges. The ranks can be viewed from either applicants' side or posts' side. Therefore, we describe our algorithm
for addition of an applicant, but the same can be used for addition of a post. The same is true for
deletion of a vertex. Deletion of an applicant or post involves deleting a vertex, along with its incident
edges. Hence the same algorithm applies to deletion of both applicants and posts.

\vspace{-4mm}

\section{The Algorithm}\label{sec:algo}
We describe the update algorithm here. Throughout this discussion, we assume that $G$ is an
instance of rank-maximal matchings and $H$ is an updated instance, where an update could be
addition or deletion of an edge or a vertex. We discuss each of these updates separately.
%and the
%algorithm has a separate module for each type of update.

As described in Section \ref{sec:preproc}, we first run Algorithm \ref{algo:Irving} on $G$ and compute a rank-maximal matching $M$ in $G$. We also store the information regarding each vertex and edge
as described in Section \ref{sec:preproc}. In the subsequent discussion, we assume that we have
the reduced graphs $G'_i$ for each rank $1\leq i \leq r$, which can be obtained in linear-time from the stored
information as proved in Lemma \ref{lem:recon}.
%
%We emulate the execution of Algorithm \ref{algo:Irving} on $H$ by using
%the reduced graphs $G'_i$ for $G$ for $1\leq i\leq r$ and construct the reduced graphs $H'_i$ for $H$.
%We also update $M$ to get a rank-maximal matching $M'$ in $H$. The algorithm iterates over each
%stage $1\leq j\leq r$. Let $M_i$ be the rank-maximal matching in $G'_i$ obtained by taking
%edges of ranks only up to $i$ from $M$.

\subsection{Addition of a vertex: }We describe the procedure for addition of a vertex in terms of addition of an applicant. Addition of a post is analogous as explained in Section \ref{sec:preproc}. A description of the vertex-addition algorithm is given below and then we prove
its correctness. The pseudocode is given in Appendix. 
%\vspace{-0.2in}
%\subsubsection{Description of vertex-addition algorithm:}

{\bf Description of vertex-addition algorithm:}
Let $a$ be a new applicant to be added to the instance $G$. Let $E_a$ be the
set of edges along with their ranks, that correspond to the preference list of $a$. Thus the new
instance is $H=((\A\cup \{a\})\cup\posts, E\cup E_a)$. The update algorithm starts from $G'_1$, adds
edges of rank $1$ from $E_a$ to $G'_1$ to get $H'_1$ and then updates $M$ and $H'_1$ as follows:

{\bf Initialization: }$S,T=\emptyset$. These sets are used later as described below.

The following cases arise while updating $H'_1$:
\begin{description}
\item[Case $1$: Each rank $1$ post $p$ of $a$ is odd in $G'_1$: ]
Then $H'_1$ is same as $G'_1$, along with $a$ and its rank $1$ edges added.

\item[Case $2$: No rank $1$ post of $a$ is even but some post is unreachable in $G'_1$: ] Update the labels $\EE,\odd,\un$.\footnote{In Irving et al.'s algorithm, these labels are called $\EE_1,\odd_1,\un_1$. We omit the subscripts for the sake of bravity. The subscripts are clear from the stage under consideration.} Add those applicants whose 
label changes from $\un$ to $\EE$ to the set $S$, as they need to get higher rank edges in subsequent stages. Note that their higher rank edges are deleted by Algorithm \ref{algo:Irving} as they become
unreachable in $G'_1$. Thus $S$ always stores the vertices which need to get higher rank edges in
subsequent iterations.

\item[Case $3$: A rank $1$ post $p$ of $a$ is even in $G'_1$: ]
Then there is an augmenting path from $a$ to $p$ in $H'_1$. Find it and augment $M_1$ to get a
rank-maximal matching $M'_1$ in $H'_1$.
Recompute the $\EE,\odd,\un$ labels. Delete higher rank edges on those vertices whose labels
change from $\EE$ in $G'_1$ to $\un$ in $H'_1$.
\end{description}
Delete $\odd\odd$ and $\odd\un$ edges if present.
Add those vertices to $T$ which are odd or unreachable in $H'_1$. These are precisely those
vertices that will not get higher rank edges in any subsequent iteration even if they become even in
one such iteration.

For each subsequent stage $i>1$, the algorithm proceeds as follows:
\begin{enumerate}
\item Start with $H'_i=G'_i$. Add $a$ and its undeleted edges up to rank $i$ to $H'_i$.
\item If there are applicants in the set $S$ as described in Case $2$ above, add edges of rank $i$ incident on them to $H'_i$.
\item Start with a matching $M'_i$ in $H'_i$ such that $M'_i$ has all the edges of $M'_{i-1}$
and those rank $i$ edges of $M_i$ which are not incident on any vertex matched in $M'_{i-1}$.
\item Check if there is an augmenting path in $H'_i$ with respect to $M'_i$. If so, augment $M'_i$.
\item Recompute the labels $\EE,\odd,\un$. 
\item Delete higher rank edges on those vertices whose labels change from $\EE$ to $\un$ or $\odd$.
Remove such vertices from $S$ if they are present in $S$.
\item Delete $\odd\odd$ or $\odd\un$ edges, if present. Now we have the final updated reduced
graph $H'_i$.
\item Add those vertices from $V\setminus T$ to $S$ whose labels change from $\un$ or $\odd$ to $\EE$. Add those vertices to $T$ which are odd or unreachable in $H'_i$.
\end{enumerate}
The algorithm stops when there are no more edges left in $H$. Figure \ref{fig:add} shows an example
of the various cases considered above.
\begin{figure}[t]
\begin{center}
\begin{minipage}{0.32\textwidth}
%\begin{subfigure}
\begin{center}
\scalebox{1.0}{\begin{tikzpicture}[
roundnode/.style={circle, draw=black!100,  fill=black, inner sep=0pt, minimum size=4pt},
squarednode/.style={circle, draw=black!100, fill=black, inner sep=0pt, minimum size=4pt},
]

%Nodes
\node[roundnode]      (maintopic)            [label=left:$a_1$] {};
\node[roundnode]        (secondcircle)       [below=5mm of maintopic, label=left:$a_2$] {};
\node[roundnode]      (thirdcircle)           [below=5mm of secondcircle, label=left:$a_3$] {};
\node[roundnode]        (fourthcircle)       [below=5mm of thirdcircle, label=left:$a_4$] {};
%Nodes
\node[squarednode]      (firstsq)  [right=of maintopic,label=right:$p_1$] {};
\node[squarednode]        (secondsq)       [right=of secondcircle,label=right:$p_2$]{};
\node[squarednode]      (thirdsq)       [right=of thirdcircle,label=right:$p_3$]{};
\node[squarednode]        (fourthsq)       [right=of fourthcircle,label=right:$p_4$]{};

%Lines
\draw[-,very thick,dashed] (maintopic.east) -- (firstsq.west);
%\draw[-] (secondcircle.east) -- (firstsq.west);
%\draw[-] (thirdcircle.east) -- (firstsq.west);
\draw[-] (fourthcircle.east) -- (firstsq.west);
\draw[-] (secondcircle.east) -- (secondsq.west);
\draw[-,very thick, dashed] (secondcircle.east) -- (thirdsq.west);
\draw[-,very thick, dashed] (thirdcircle.east) -- (secondsq.west);
\draw[-] (thirdcircle.east) -- (fourthsq.west);
\end{tikzpicture}}
\vspace{-0.1in}
\begin{eqnarray*}
&(i)&\\
a_1 &:& p_1 \\
a_2 &:& p_1,(p_2,p_3) \\
a_3 &:& p_1,(p_2,p_4)\\
a_4 &:& p_1
\end{eqnarray*}

\end{center}
%\end{subfigure}
\end{minipage}
\begin{minipage}{0.32\textwidth}
\begin{center}
\scalebox{1.0}{\begin{tikzpicture}[
roundnode/.style={circle, draw=black!100,  fill=black, inner sep=0pt, minimum size=4pt},
squarednode/.style={circle, draw=black!100, fill=black, inner sep=0pt, minimum size=4pt},
]

%Nodes
\node[roundnode]      (maintopic)                              [label=left:$a_1$] {};
\node[roundnode]        (secondcircle)       [below=5mm of maintopic,label=left:$a_2$]{};
\node[roundnode]      (thirdcircle)           [below=5mm of secondcircle,label=left:$a_3$]{};
\node[roundnode]        (fourthcircle)       [below=5mm of thirdcircle,label=left:$a_4$]{};
%Nodes
\node[squarednode]      (firstsq)  [right=of maintopic, label=right:$p_1$]{};
\node[squarednode]        (secondsq)       [right=of secondcircle,label=right:$p_2$]{};
\node[squarednode]      (thirdsq)       [right=of thirdcircle,label=right:$p_3$]{};
\node[squarednode]        (fourthsq)       [right=of fourthcircle,label=right:$p_4$]{};

%Lines
\draw[-, very thick, dashed] (maintopic.east) -- (firstsq.west);
\draw[-] (secondcircle.east) -- (firstsq.west);
%\draw[-] (thirdcircle.east) -- (firstsq.west);
\draw[-] (fourthcircle.east) -- (firstsq.west);
\draw[-,very thick, dashed] (thirdcircle.east) -- (secondsq.west);
\draw[-] (thirdcircle.east) -- (thirdsq.west);
\draw[-] (thirdcircle.east) -- (fourthsq.west);
\end{tikzpicture}}
\vspace{-0.1in}
\begin{eqnarray*}
&(ii)&\\
a_1 &:& p_1\\
a_2 &:& p_1\\
a_3 &:& p_1,(p_2,p_3,p_4)\\
a_4 & :& p_1
\end{eqnarray*}

\end{center}
\end{minipage}
\begin{minipage}{0.32\textwidth}
\begin{center}
\scalebox{1.0}{\begin{tikzpicture}[
roundnode/.style={circle, draw=black!100,  fill=black, inner sep=0pt, minimum size=4pt},
squarednode/.style={circle, draw=black!100, fill=black, inner sep=0pt, minimum size=4pt},
]

%Nodes
\node[roundnode]      (maintopic)                              [label=left:$a_1$]{};
\node[roundnode]        (secondcircle)       [below=5mm of maintopic, label=left:$a_2$]{};
\node[roundnode]      (thirdcircle)           [below=5mm of secondcircle,label=left:$a_3$]{};
\node[roundnode]        (fourthcircle)       [below=5mm of thirdcircle,label=left:$a_4$]{};
%Nodes
\node[squarednode]      (firstsq)  [right=of maintopic,label=right:$p_1$]{};
\node[squarednode]        (secondsq)       [right=of secondcircle,label=right:$p_2$]{};
\node[squarednode]      (thirdsq)       [right=of thirdcircle,label=right:$p_3$]{};
\node[squarednode]        (fourthsq)       [right=of fourthcircle,label=right:$p_4$]{};

%Lines
\draw[-, very thick, dashed] (maintopic.east) -- (firstsq.west);
%\draw[-] (secondcircle.east) -- (firstsq.west);
%\draw[-] (thirdcircle.east) -- (firstsq.west);
\draw[-] (fourthcircle.east) -- (firstsq.west);
\draw[-,very thick, dashed] (secondcircle.east) -- (secondsq.west);
\draw[-,very thick, dashed] (thirdcircle.east) -- (thirdsq.west);
\draw[-] (fourthcircle.east) -- (thirdsq.west);
\draw[-] (thirdcircle.east) -- (fourthsq.west);
\end{tikzpicture}}
\vspace{-0.1in}
\begin{eqnarray*}
&(iii)&\\
a_1 &:& p_1 \\
a_2 &: & p_1, p_2\\
a_3 & : & p_1,(p_3,p_4)\\
a_4 &: & p_1,p_3 
\end{eqnarray*}
\end{center}
\end{minipage}
\caption{Example of status change of nodes after addition of applicant $a_4$. Dashed lines indicate a rank-maximal matching before addition of $a_4$. In $(i)$, $a_1,p_1$ are unreachable before adding $a_4$. After adding $a_4$, $p_1$ becomes odd while $a_1$ becomes even. 
In $(ii)$,
there is no status change after adding $a_4$. In $(iii)$, there is an augmenting path $a_4,p_3,a_3,p_4$ after adding $a_4$. Augmentation makes all the nodes unreachable. Preference list for each figure is shown below the figure. Note that some edges on $p_1$ are deleted because they are $\odd\odd$ or $\odd\un$ edges.}\label{fig:add}
\end{center}
\end{figure}
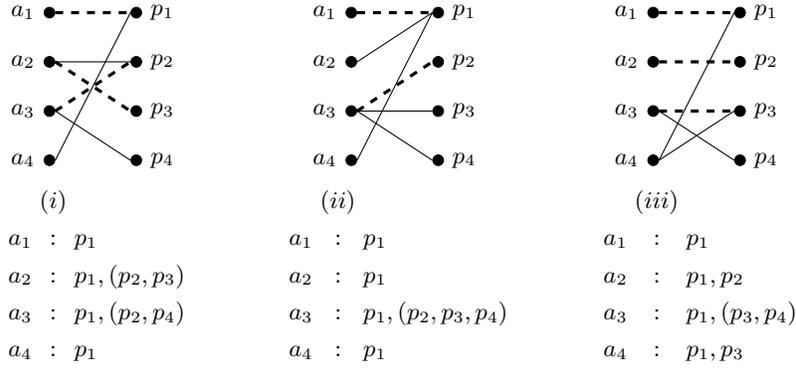
\vspace{-0.2in}

\subsubsection{Analysis of the vertex-addition algorithm:}
Recall the notation that $G$ is the given instance and $H$ is the instance obtained by adding an
applicant $a$ along with its incident edges. Moreover, $G_i$ and $H_i$ are subgraphs of $G$ and
$H$ respectively, consisting of edges up to rank $i$ respectively from $G$ and $H$.
Also $G'_i$ is the reduced graph corresponding to stage $i$ of an execution of Algorithm
\ref{algo:Irving} on $G$ whereas $H'_i$ is the graph of stage $i$ for $H$ constructed by
the vertex-addition algorithm. In Theorem \ref{thm:add}, we prove that $H'_i$ is indeed the 
reduced graph that would be constructed by an execution of Algorithm \ref{algo:Irving} on $H$.

The following Lemma is useful in analyzing the running time of the algorithm. 
 It proves that there can be at most one new augmenting path
at any stage $i$ with respect to $M_i$ in $H'_i$. Recall that $M$ is a rank-maximal matching in $G$ and $M_i$ is the subset of $M$ consisting of edges of rank only up to $i$. We give the proof in Appendix.
\begin{lemma}\label{lem:aug}
At each stage $i$, $|M_i|\leq |M'_i|\leq |M_i|+1$.
Thus, for any stage $i$, there can be at most one augmenting path
with respect to $M_i$ in $H_i$.
\end{lemma}

Correctness of the algorithm is given by the following theorem, we prove its base case here, the full proof appears in Appendix.
\begin{theorem}\label{thm:add}
Algorithm~\ref{algo:add} correctly updates the rank-maximal matching and the reduced graphs.
Moreover, it runs in time $O(r(m+n))$.
\end{theorem}
\begin{proof}
We prove this by induction on ranks. Thus we prove that, if the stage-wise reduced graphs
are updated correctly up to stage $i-1$, then the algorithm correctly constructs $H'_i$,
and gives a rank-maximal matching $M'_i$ in $H_i$. 

As base case, consider the graph $G'_1$ and let $a$ be added to $G'_1$ along with his rank $1$ edges.

\begin{description}
\item[Case $1$: Each rank $1$ post $p$ of $a$ is odd in $G'_1$:] Then $p$ has an alternating path from an unmatched applicant 
in $G'_1$. Addition of $a$ only creates one more such path, so there is no augmentation and no
change of labels. Applicant $a$ remains unmatched and hence even. This can be checked in $O(1)$
time for each post using the information stored at the preprocessing stage. This case is considered
in line \ref{add:odd} of Algorithm~\ref{algo:add}. Thus $H'_1$ is same as $G'_1$ with $a$ and its 
rank $1$ edges added. Also $M'_1=M_1$. 

\item[Case $2$: A rank $1$ post $p$ of $a$ is unreachable in $G'_1$ but none is even: ] Since $p$
is unreachable, there is no alternating path to $p$ from an unmatched applicant or post in $G'_1$. Addition of 
the edge $(a,p)$ creates such a path. Hence the label of $p$ changes from $\unreach$ in $G'_1$ to $\odd$ in $H'_1$. The label on the matched partner $M(p)$ of $p$ in $M$ then changes from $\un$ to
$\EE$. There could be other applicants and posts which are unreachable in $G'_1$ but have alternating
paths from $p$ that use the edge $(p,M(p))$. Such applicants and posts now have respectively an even and odd
length alternating path from $a$ and hence their labels change from $\un$ to $\EE$ and $\odd$
respectively.
Note that these
alternating paths are considered with respect to the existing matching $M_1$, and $M'_1=M_1$. 

Consider the applicants whose labels change from $\unreach$ to $\EE$.
As these applicants are unreachable in $G'_1$, Algorithm~\ref{algo:Irving} must have deleted
their higher rank edges from $G$. These edges need to be added back as they have become even now. We include them in the set $S$ so that such edges can be added at the respective stages. 
%This is done in lines \ref{add:unrecomp} and \ref{add:mark}
%of Algorithm~\ref{algo:add} and edges on applicants in $S$ are added in line \ref{add:update}.

\item[Case $3$: Applicant $a$ has a rank $1$ post $p$ which is even in $G'_1$:] Then $p$ has an 
alternating path from some unmatched post $q$ (possibly $p=q$). 
This, along with $a$, now forms an augmenting
path and $M_1$ needs to be augmented. This path can be found in $O(m+n)$ time by a BFS or DFS
from $a$ and $M_1$ is augmented to get $M'_1$ in the same time.

This augmentation leads to changing $q$ from an unmatched to matched post. Now $p$ 
may not have an alternating path from an unmatched post. If this happens, $p$ becomes unreachable.
Other posts on the alternating path from $q$ to $p$, if any, also become unreachable and their
higher rank edges need to be deleted. Their 
corresponding matched applicants, that were odd earlier, also become unreachable. This needs a recomputation of $\EE,\odd,\un$ labels. Also, higher rank edges on those posts whose labels change from $\EE$ to $\un$ need to be deleted from $H$. %This is dealt with in lines \ref{add:augbeg} to \ref{add:augend} of Algorithm~\ref{algo:add}. 
Note that if $p$ has an alternating path from 
an unmatched post in $H'_1$ with respect to $M'_1$, then there is no change of labels after augmentation of $M_1$. 

If there are new $\odd\odd$ or $\odd\un$ edges in $H'_1$, they need to be deleted. 
This completes the base case. The induction step is similar, and is given in Appendix.

\end{description}

\end{proof}
%Then $p$ has an alternating
%path in $G'_1$ from an unmatched applicant. Addition of $a$ only creates one more such path. So
%$M$ remains unchanged. The labels $\EE,\odd,\un$ on all the vertices also remain the same. As
%$a$ remains unmatched, it remains even. In this case, $H'_1$ is same as $G'_1$ with $a$ and its rank $1$ edges added.
%\begin{itemize}
%\item {\bf Case $2$: No rank $1$ post of $a$ is even but one or more of them are unreachable in $G'_1$: }Let $p$ be a rank $1$ post of $a$ such that $p$ is unreachable in $G'_1$.
%Thus $p$ has no alternating path in $G'_1$ from an unmatched applicant or post. In $H'_1$, $p$ 
%has an alternating path from $a$, an unmatched applicant. Hence the label on $p$ changes from $\un$
%to $\odd$. As $p$ does not have an alternating path from an unmatched post in $G'_1$, there is no
%augmenting path from $a$. However, $M(p)$, the matched partner of $p$ in $M$,  has an alternating path from
%$a$ in $H'_1$. So the label of $M(p)$ changes from $\un$ in $G'_1$ to $\EE$ in $H'_1$. There could be applicants and posts which have alternating paths from $p$ in $G'_1$ beginning with the edge 
%$(p,M(p)$. Such vertices have an alternating path from $a$ in $H'_1$. So the labels on those applicants
%and posts change from $\un$ to $\EE$ and $\odd$ respectively. An applicant whose label is $\un$ in 
%$G'_1$ has no edges of rank $i>1$ in any $G'_j$, $j>1$, since higher rank edges on odd or unreachable vertices are deleted in Algorithm \ref{algo:Irving}. However, in $H'_1$, if the label of 
%\end{itemize}
%\newpage
 
\vspace{-7mm}
\subsection{Deletion of a vertex}\label{sec:del}
Let an applicant $a$ be deleted from the instance. The case of deletion of a post $p$ is analogous, as explained in Section \ref{sec:preproc}. Let $G$ be the given instance and $H$ be the updated instance. Thus $H=(\A\setminus \{a\} \cup \posts, E\setminus E_a)$ where $E_a$ is the set of edges incident on $a$. Let $M$ be a rank-maximal matching in $G$. Also assume that the preprocessing step is executed on $G$ and the information as mentioned in Section \ref{sec:preproc} is stored.
\vspace{-5mm}

\subsubsection{Description of the vertex-deletion algorithm}
If $a$ is not matched in $M$, then $M$ clearly remains rank-maximal in $H$, although the reduced
graphs $H'_i$ could differ a lot from the corresponding reduced graphs $G'_i$ for each $i$. We describe the algorithm below, the pseudocode is given in Appendix.
%For the ease of 
%description, if $a$ is not matched in $M$, introduce a new post $p_a$ in $a$'s preference list at
%rank $r+1$ where $r$ is the maximum rank used in $M$. Include $(a,p_a)$ edge in $M$. So here 
%onwards, we can assume that $a$ is always matched.

{\bf Initialization: }$S,T=\emptyset$. These sets will be used later, as given in the following description.

\begin{description}
\item[Case (I): $a$ is matched in $M$: ]Let $j$ be the rank of the matched edge in $M$ incident on $a$ and Let $M(a)=p$. Thus $a$ remains
even in the execution of Algorithm \ref{algo:Irving} on $G$ at least for $j$ iterations. The algorithm now
works as follows:

For each rank $i$ from $1$ to $j-1$, initialize $H'_i=G'_i$ and $M'_i=M_i$. 
Delete edges of rank up to $i$ incident
on $a$ from $H'_i$. Recompute the labels $\EE,\odd,\un$. Delete from $H$ the edges of rank $>i$ on those applicants whose label changes from $\EE$ to $\un$. This is the final reduced graph $H'_i$. 
Add odd and unreachable vertices from $H'_i$ to $T$. The set $T$ contains those vertices that will
not get higher rank edges at later stages even if their label changes to $\EE$.

Now we come to the rank $j$ at which $a$ is matched in $M$.
Initialize $H'_j=G'_j$ and $M'_j=M_j\setminus \{(a,p)\}$. Delete edges incident on $a$ from $H'_j$.
The following cases arise:
\begin{description}
\item[Case $1$: $p$ is odd in $G'_j$: ] Find an augmenting path in $H'_j$ with respect to $M'_j$ starting at $p$. Augment $M'_j$ along this path.
Recompute the labels. Delete from $H$ the edges of rank $>j$ incident on those applicants whose labels change from $\EE$ in $G'_j$ to $\un$ in $H'_j$.
\item[Case $2$: $p$ is unreachable in $G'_j$: ]Recompute the labels $\EE,\odd,\un$. Include those posts to $S$ whose label changes from $\un$ to $\EE$. These posts need to get edges of rank $>j$ in subsequent iterations.
\item[Case $3$: $p$ is even in $G'_j$: ]Recompute the labels $\EE,\odd,\un$ in $H'_j$. Remove higher rank
edges on those posts whose labels change from $\EE$ to $\un$.
\end{description}
Add the odd and unreachable vertices from $H'_j$ to $T$. Remove such vertices from $S$, if they are present in $S$. These are the vertices that will not
get higher rank edges even if they get the label $\EE$ at a later stage.

For each rank $i$ from $j+1$ to $r$, initialize $H'_i=G'_i$, except for $a$ and its incident edges.
Add edges of rank $i$ on posts in $S$. Initialize $M'_i=M'_{i-1}\cup\textrm{ set of those edges in }M_i\textrm{ which are disjoint from the edges in }M'_{i-1}$. Look for an augmenting path, and augment $M'_i$ if an augmenting path is found. Recompute the labels $\EE,\odd,\un$. Update $S$ and $T$ as
mentiond above.
\item[Case (II): $a$ is unmatched in $M$:] The algorithm involves iterating over $i=1$ to $r$ and computing the reduced graphs $H'_i$ as follows: Start with $H'_i=G'_i$,
deleting $a$ and its incident edges from $H'_i$, add rank $i$ edges on vertices in $S$, recompute the labels $\EE,\odd,\un$, include those 
vertices from $V\setminus T$ into set $S$ whose labels change from $\un$ to $\EE$. Add vertices
with $\odd$ or $\un$ labels to $T$. Delete higher rank edges on the vertices whose labels are $\odd$
or $\un$.
\end{description}

The correctness of the vertex-deletion algorithm is given by the theorem below. The proof and an example (Figure \ref{fig:del} appear in Appendix. 

\begin{theorem}\label{thm:del}
Algorithm~\ref{algo:del} correctly updates the rank-maximal matching $M$ on deletion of an applicant.
Moreover, it takes time $O(r(m+n))$.
\end{theorem}

\vspace{-0.15in}
\subsection{Addition and deletion of an edge}\label{sec:edge}
Modules similar to those for vertex-addition and vertex-deletion can be written for addition and deletion of an edge, which would have time complexity $O(r(m+n))$ each. However, both edge-addition and edge -deletion can be 
performed as a vertex-deletion followed by vertex-addition, achieving the same running time $O(r(m+n))$. We explain this here. To add an edge $(a,p)$, one can first delete applicant $a$ using the vertex-deletion algorithm thereby deleting all the edges $E_a$ incident on $a$, and then the applicant $a$ is added back along with the edge-set $E_a\cup\{(a,p)\}$. Similarly, deletion of an edge $(a,p)$ can be
carried out by first deleting the applicant $a$ along with the set of edges $E_a$ incident on $a$ and then
adding back $a$ along with the edge-set $E_a\setminus\{(a,p)\}$. 
It is clear that both edge-addition and edge-deletion can thus be carried out in $O(r(m+n))$ time.
\vspace{-0.15in}
\section{Discussion}\label{sec:disc}
	In this paper, we give an $O(r(m+n))$ algorithm to update a rank-maximal matching when vertices or edges are added
and deleted over time. Independent of our work, \cite{Ghosal17} give an algorithm for vertex addition and deletion that runs in $O(m)$ time using similar techniques.

In \cite{GNN14}, a switching 
graph characterization of rank-maximal matchings has been developed, which has found several applications. A natural question to ask
is whether this characterization is also useful in dynamic setting. However, a switching graph
is based on the {\em reduced graph} computed by Irving et al.'s algorithm, which is a subgraph
of the input graph. Addition or deletion of a vertex can change this subgraph and hence the
switching graph significantly. Therefore it is not immediate whether the switching graph characterization
can help in dynamic setting. It is an interesting question to explore.

\noindent{\bf Acknowledgement: }We thank anonymous reviewers for their comments on an
earlier version of this paper. We thank Meghana Nasre for helpful discussions.
 
\bibliographystyle{abbrv}
\bibliography{references}
\newpage
\appendix\label{sec:app}
\section{Example for addition of an edge}
We give an example to show that addition of an edge can change the rank-maximal matching
by $\Omega(n)$ edges. 

Let the given instance be as follows:
%\begin{minipage}{0.49\textwidth}
\begin{eqnarray*}
a_1 &:&   p_1\\
a_2 &:&  p_5,  p_1,  p_2\\
a_3 &:&  p_5,  p_6,  p1,   p_2,    p_3\\
a_4 &:&  p_5,  p_6,  p_1,   p_7,    p_2,  p_3,   p_4\\
a_5 &:&  p_5\\
a_6 &:&  p_6,   p_8\\
a_7& :& p_7
\end{eqnarray*}
%\end{minipage}
%\begin{mini
The instance has only one rank-maximal matching given by  
$M=\{(a_1, p_1), (a_2, p_2), (a_3, p_3),(a_4, p_4),(a_5, p_5),(a_6, p_6),(a_7, p_7)\}$

Now consider addition of an edge $(a_1,p_8)$ of rank $1$, so that the instance becomes
\begin{eqnarray*}
a_1&:&  (p_1,p_8)\\
a_2&: & p_5,  p_1,  p_2\\
a_3&:&  p_5,  p_6,  p_1,   p_2,    p_3\\
a_4&: & p_5,  p_6,  p_1,   p_7,    p_2,   p_3,   p_4\\
a_5&:& p_5\\
a_6&: &p_6 ,  p_8\\  
a_7&: &p_7
\end{eqnarray*}	

This new instance also admits only one rank-maximal matching $M'$ given by
$M'=\{(a_1, p_8),(a_2, p_1),(a_3, p_2),(a_4, p_3),(a_5, p_5),(a_6, p_6),(a_7, p_7)\}$

Note that $M$ and $M'$ differ by $4$ edges, which is more than half the size of $M$ or $M'$.
The example can be easily scaled for any number of applicants.

\section{Details of vertex-addition}\label{sec:appadd}
\begin{algorithm}[!h]
\begin{algorithmic}[1]
%\STATE{If the addition of applicant $a$ changes the labels on some applicants from $\unreach$ to
%$\even$ at any stage, we mark such applicants. These applicants will get back their higher rank edges  in later stages, which were deleted in Algorithm~\ref{algo:Irving}.}
\STATE $S=\emptyset, T=\emptyset$
\FOR{each rank $i$ from $1$ until $a$ is matched}
\STATE\label{add:update} Update $G'_i$ to get $H'_i$: If there are vertices in $S$ (added in step \ref{add:mark} of previous iteration), add rank $i$ edges incident on them. These are 
the applicants which changed from $\unreach$ to $\EE$ in one of the previous stages. Update
the $\odd,\unreach,\EE$ labels.
\STATE Add edges between $a$ and those of his rank $i$
posts which do not become odd or unreachable in $H'_j$ for any $j<i$. 
\IF{All of $a$'s rank $i$ posts are odd in $G'_i$}
\STATE\label{add:odd} There is no change in labels, $a$ remains even. Do nothing.
\ELSIF{One or more of $a$'s rank $i$ posts are unreachable in $G'_i$ and no rank $i$ post of $a$ is even in $G'_i$}
\STATE\label{add:unrecomp} Recompute the labels. 
\STATE \label{add:mark}Now some unreachable posts become odd and corresponding unreachable
applicants become even. Include these applicants to $S$ for addition of higher rank edges later if they
are not present in $T$.
\ELSIF{One of $a$'s rank $i$ posts is even in $G'_i$}
\STATE \label{add:augbeg}Augment $M_i$ by finding an augmenting
path from $a$. Call this matching $M'_i$. 
\STATE Recompute the labels $\EE,\odd,\unreach$. 
\COMMENT{/* Now some even posts may become unreachable. Some odd applicants may
become unreachable. The applicant $a$ will be odd or unreachable.*/} 
\STATE Delete higher rank edges on the posts whose labels change from $\EE$ to $\unreach$. 
 \STATE Delete $\odd\unreach$ and $\odd\odd$ edges from $H'_i$. 
\STATE\label{add:augend}Remove those vertices from $S$ become $\odd$ or $\unreach$ now. This is the updated graph $H'_i$.
\STATE Add odd or unreachable vertices in $H'_i$ to $T$.
\ENDIF
\ENDFOR
\STATE Now $a$ is matched to some post $p$ by a rank $i$ edge in $M'_i$.

\FOR{each rank $j$=$i+1$ to $r$}
\STATE\label{add:updatebeg}Start with $H'_j=G'_j$. Add edges of rank $j$ incident on vertices in $S$
to $H'_j$.
\STATE Update $M_j$ to reflect the changes that were made at earlier stages, and call it $M'_j$. 
\STATE Search for an augmenting path in $H'_j$ with respect to $M'_j$.Augment $M'_j$ if an augmenting path is found.
\STATE Relabel vertices and delete $\odd\odd$ and $\odd\unreach$ edges from $H'_j$. 
\COMMENT{/*This is the final $H'_j$. Also $M'_j$ is the rank-maximal matching in $H_j$.*/}
\STATE\label{add:updateend} Remove the vertices from $S$ which are now odd or unreachable.
\STATE Add odd or unreachable vertices to $T$. Delete higher rank edges on them from $H$.

%\ENDIF
\ENDFOR
\end{algorithmic}
\caption{Update algorithm for addition of a new applicant $a$}\label{algo:add}
\end{algorithm}

\begin{appendix-lemma}{\ref{lem:aug}}
At each stage $i$, $|M_i|\leq |M'_i|\leq |M_i|+1$.
Thus, for any stage $i$, there can be at most one augmenting path
with respect to $M_i$ in $H_i$.
\end{appendix-lemma}
\begin{proof}
Recall from invariant $(I3)$ of Algorithm \ref{algo:Irving} mentioned in Section \ref{sec:prelim}
that $M_i$ and $M'_i$ are the rank-maximal matchings in $G_i$ and 
$H_i$ respectively. Here $G_i$ and $H_i$ are the instances $G$ and $H$ 
with only the edges of ranks $1$ to $i$ present.

Consider $M_i\oplus M'_i$, which is the set of edges present in exactly one of the two matchings.
This is a collection of vertex-disjoint
paths and cycles. Each path that does not contain the new applicant $a$, and each
cycle must have the same number of edges of each rank from $M_i$ and 
$M'_i$. Otherwise we can obtain a matching which has a better signature
than either $M_i$ or $M'_i$, which contradicts the rank-maximality of 
both the matchings in $G_i$ and $H_i$ respectively. At most one path
can have the new applicant $a$ as one end-point. This path can contain at most one
more edge of $M'_i$ than that of $M_i$. This proves the first part.

To see that there can be an augmenting path at multiple stages, consider 
the case where a
post $p$ gets matched to the new applicant $a$ at stage $i$. If $p$ is matched
to an applicant $b$ in $M$ and the edge $(b,p)$ has rank $j$ such that $j>i$, then $b$
is matched in $G'_j$ but not in $H'_j$. Hence there can possibly be
a new augmenting path in $H'_j$ starting at $b$.  
\qed\end{proof} 
\begin{appendix-theorem}{\ref{thm:add}}
Algorithm~\ref{algo:add} correctly updates the rank-maximal matching and the reduced graphs.
Moreover, it runs in time $O(r(m+n))$.
\end{appendix-theorem}
\begin{proof}
We prove this by induction on ranks. Thus we prove that, if the stage-wise reduced graphs
are updated correctly up to stage $i-1$, then the algorithm correctly constructs $H'_i$,
and gives a rank-maximal matching $M'_i$ in $H_i$. 

As base case, consider the graph $G'_1$ and let $a$ be added to $G'_1$ along with his rank $1$ edges.

\begin{description}
\item[Case $1$: Each rank $1$ post $p$ of $a$ is odd in $G'_1$:] Then $p$ has an alternating path from an unmatched applicant 
in $G'_1$. Addition of $a$ only creates one more such path, so there is no augmentation and no
change of labels. Applicant $a$ remains unmatched and hence even. This can be checked in $O(1)$
time for each post using the information stored at the preprocessing stage. This case is considered
in line \ref{add:odd} of Algorithm~\ref{algo:add}. Thus $H'_1$ is same as $G'_1$ with $a$ and its 
rank $1$ edges added. Also $M'_1=M_1$. 

\item[Case $2$: A rank $1$ post $p$ of $a$ is unreachable in $G'_1$ but none is even: ] Since $p$
is unreachable, there is no alternating path to $p$ from an unmatched applicant or post in $G'_1$. Addition of 
the edge $(a,p)$ creates such a path. Hence the label of $p$ changes from $\unreach$ in $G'_1$ to $\odd$ in $H'_1$. The label on the matched partner $M(p)$ of $p$ in $M$ then changes from $\un$ to
$\EE$. There could be other applicants and posts which are unreachable in $G'_1$ but have alternating
paths from $p$ that use the edge $(p,M(p))$. Such applicants and posts now have respectively an even and odd
length alternating path from $a$ and hence their labels change from $\un$ to $\EE$ and $\odd$
respectively.
Note that these
alternating paths are considered with respect to the existing matching $M_1$, and $M'_1=M_1$. 

Consider the applicants whose labels change from $\unreach$ to $\EE$.
As these applicants are unreachable in $G'_1$, Algorithm~\ref{algo:Irving} must have deleted
their higher rank edges from $G$. These edges need to be added back as they have become even now. We include them in the set $S$ so that such edges can be added at the respective stages. This is done in lines \ref{add:unrecomp} and \ref{add:mark}
of Algorithm~\ref{algo:add} and edges on applicants in $S$ are added in line \ref{add:update}.

\item[Case $3$: Applicant $a$ has a rank $1$ post $p$ which is even in $G'_1$:] Then $p$ has an 
alternating path from some unmatched post $q$ (possibly $p=q$). 
This, along with $a$, now forms an augmenting
path and $M_1$ needs to be augmented. This path can be found in $O(m+n)$ time by a BFS or DFS
from $a$ and $M_1$ is augmented to get $M'_1$ in the same time.

This augmentation leads to changing $q$ from an unmatched to matched post. Now $p$ 
may not have an alternating path from an unmatched post. If this happens, $p$ becomes unreachable.
Other posts on the alternating path from $q$ to $p$, if any, also become unreachable and their
higher rank edges need to be deleted. Their 
corresponding matched applicants, that were odd earlier, also become unreachable. This needs a recomputation of $\EE,\odd,\un$ labels. Also, higher rank edges on those posts whose labels change from $\EE$ to $\un$ need to be deleted from $H$. This is dealt with in lines \ref{add:augbeg} to \ref{add:augend} of Algorithm~\ref{algo:add}. 
Note that if $p$ has an alternating path from 
an unmatched post in $H'_1$ with respect to $M'_1$, then there is no change of labels after augmentation of $M_1$. 

If there are new $\odd\odd$ or $\odd\un$ edges in $H'_1$, they need to be deleted. 

\end{description}
The algorithm also stores those vertices which are odd or unreachable in $H'_1$ in a set $T$.
These are precisely those vertices that will not be added
to $S$ at any point and hence will not get higher rank edges later. 
Thus at the end of stage $1$, we have the graph $H'_1$ exactly same as what would be 
given by executing Algorithm \ref{algo:Irving} on $H$. We also have a maximum 
matching $M'_1$ in $H_1$, which is trivially rank-maximal when edges up to rank $1$ 
are considered.

We now come to the inductive part. Assume that the algorithm has correctly computed $H'_j$ for
$1\leq j < i$. We show that the algorithm then correctly computes $H'_i$ and $M'_i$. 

\begin{description}
\item[Initialization]
The algorithm starts from
$H'_i=G'_i$ and the matching $M'_i$ in $H_i$ is initialized to $M'_{i-1}\cup \textrm{set of those edges
in }M_i\textrm{ that are vertex-disjoint from edges in }M'_{i-1}$. Note that there could be a vertex 
that is matched in $M'_{i-1}$ but not in $M_{i-1}$, and possibly matched in $M_i$. Thus the initial
matching $M'_i$ is same as $M_i$ except for the updates performed in earlier stages. 

Recall that $S$ is the set of vertices which do not have rank $i$ edges in $G'_i$ but need to get 
rank $i$ edges in $H'_i$. The algorithm adds rank $i$ edges on such vertices. It also adds applicant
$a$ and its undeleted edges to $H'_i$. 
\item[Checking for augmenting path:] If $a$ is still unmatched, then there could be an augmenting
path in $H'_i$ starting at $a$. Even if $a$ is matched in $M'_{i-1}$, there could still be an augmenting path in $H'_i$ with respect to $M'_i$, as explained below:

If $a$ is matched in $M'_{i-1}$, say by a rank $j\leq i-1$ edge, 
then there is also a post that is 
matched in $M'_{i-1}$ but not in $M_{i-1}$, say $q$. This is because augmentation along an augmenting path always matches an additional applicant (in this case, $a$) and an additional post (in this case $q$). However, in $M$, i.e. prior to addition of $a$, $q$ may have been matched
to some applicant $b$ at a rank $k>j$. Now $q$ is matched to $a$, so $b$ loses its matched edge
at stage $k$. This needs updating labels $\odd,\un,\EE$ at subsequent stages. Also, we need to
find an augmenting path from $b$, if any, at a later stage. This is done in lines \ref{add:updatebeg} to \ref{add:updateend}.

Note that there can be at most one augmenting path according to Lemma \ref{lem:aug}.

\item[Recomputation of labels: ]Thus, at each stage, the algorithm looks for an augmenting path and augments $M'_i$, if such
a path is found. The augmentation can lead to change of labels, and deletion of edges on those
vertices whose labels change from $\EE$ to $\un$ due to the augmentation. Also, even if there is no
augmentation, there could still be a change of labels due to addition of edges on vertices in $S$
and also due to addition of edges incident on $a$. Thus the labels need to be recomputed anyway.
The sets $S$ and $T$ are updated as mentioned in the base case above.
\end{description}
As all the possible differences between $G'_i$ to $H'_i$ are considered above, $H'_i$ is the correct
reduced graph of stage $i$. Further, by Lemma \ref{lem:aug}, $M'_i$ is a maximum matching in $H_i$
obtained by augmenting a rank-maximal matching $M'_{i-1}$ from $H'_{i-1}$. Thus $M'_i$ is 
rank-maximal in $H_i$ by correctness of Algorithm \ref{algo:Irving}.

Each of the three operations described above need $O(m+n)$ time.
Whenever the label on a vertex changes, or an edge is deleted, the stored information can be 
updated in $O(1)$ time.
Thus each stage can be updated in time $O(m+n)$, so total update time is $O(r(m+n))$.
\hfill\qed\smallskip
\end{proof}

\section{Details of vertex-deletion}\label{sec:appdel}
\begin{algorithm}[ht]
\begin{algorithmic}[1]
\STATE Let $p$ be the post to which $a$ is matched by a rank $j$ edge in $M$. 
\STATE $S=T=\emptyset$
\FOR{each rank $i=1$ to $j-1$}
\STATE\label{del:upbeg} Initialize $H'_i=G'_i$, $M'_i=M_i$. Remove $a$ and its incident edges from $H'_i$. \COMMENT{/*Certainly $a$ is even in $G'_i$./*}
\STATE Recompute the labels 
$\EE,\odd,\unreach$. 
\STATE\label{del:upend} This can change some posts from $\odd$ to $\unreach$ and their matched applicants from $\EE$ to $\unreach$. Delete higher rank edges on such applicants from $H$.
\STATE Include vertices with labels $\odd$ and $\un$ into $T$.
\ENDFOR
\STATE Initialize $H'_j=G'_j$, $M'_j=M_j\setminus \{(a,p)\}$. Remove $a$ and its incident edges from $H'_j$.
\COMMENT{This makes post $p$ unmatched and hence even.}
\IF {$p$ is odd in $G'_j$}\label{augbeg}
\STATE\label{del:oddbeg} Find an augmenting path from $p$ in $H'_j$ respect to $M'_j$ and augment $M'_j$. 
\STATE Recompute the labels $\EE,\odd,\unreach$. 
\STATE\label{del:oddend} Labels on some applicants may change from $\EE$ to $\unreach$. Delete higher rank edges on them. Labels on some posts may change from $\odd$ to $\unreach$.
\ELSIF{$p$ is unreachable in $G'_j$}
\STATE Recompute the labels $\EE,\odd,\un$ in $H'_j$ with respect to $M'_j$.
\STATE Labels on some posts including $p$ change from $\unreach$ to $\EE$. Include these posts 
into $S$ for addition of higher rank edges in later stages unless they are in $T$. 
\STATE Labels on the applicants matched to these posts change from $\unreach$ to $\odd$. 
\ELSIF{$p$ is even in $G'_j$} 
\STATE Recompute the labels $\EE,\odd,\un$ in $H'_j$ with respect to $M'_j$.
\ENDIF\label{augend}
\FOR{each rank $i=j+1$ to $r$}
\STATE Initialize $H'_i=G'_i$, 
\STATE $M'_i=M'_{i-1}\cup\textrm{ set of those edges in }M_i\textrm{ which are
disjoint from edges in }M'_{i-1}$. 
\STATE Remove $a$ and its incident edges from $H'_i$ and $M'_i$. 
Add rank $i$ edges on posts in $S$, if any.
\STATE Check for an augmenting path in $H'_i$ with respect to $M'_i$ and augment $M'_i$ if such an
augmenting path is found.
\STATE Recompute the labels $\EE,\odd,\un$ in $H'_i$ with respect to $M'_i$.
\STATE Include 
\STATE 
\ENDFOR
\end{algorithmic}
\caption{Update algorithm for deletion of an applicant $a$}\label{algo:del}
\end{algorithm}
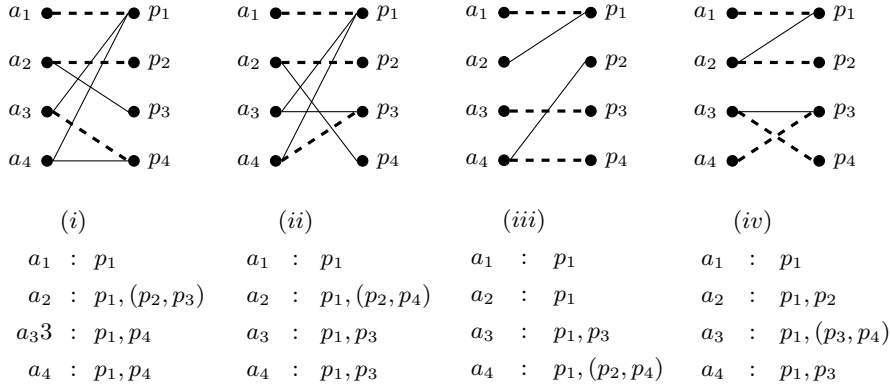
\begin{figure}[t]
\begin{minipage}{0.24\textwidth}
\scalebox{1.0}{\begin{tikzpicture}[
roundnode/.style={circle, draw=black!100,  fill=black, inner sep=0pt, minimum size=4pt},
squarednode/.style={circle, draw=black!100, fill=black, inner sep=0pt, minimum size=4pt},
]

%Nodes
\node[roundnode]      (maintopic) [label=left:$a_1$]                             {};
\node[roundnode]        (secondcircle)       [below=5mm of maintopic, label=left:$a_2$] {};
\node[roundnode]      (thirdcircle)           [below=5mm of secondcircle, label=left:$a_3$] {};
\node[roundnode]        (fourthcircle)       [below=5mm of thirdcircle, label=left:$a_4$] {};
%Nodes
\node[squarednode]      (firstsq)  [right=of maintopic, label=right:$p_1$]          {};
\node[squarednode]        (secondsq)       [right=of secondcircle, label=right:$p_2$] {};
\node[squarednode]      (thirdsq)       [right=of thirdcircle, label=right:$p_3$] {};
\node[squarednode]        (fourthsq)       [right=of fourthcircle,label=right:$p_4$] {};

%Lines
\draw[-,very thick, dashed] (maintopic.east) -- (firstsq.west);
\draw[-] (secondcircle.east) -- (thirdsq.west);
\draw[-] (thirdcircle.east) -- (firstsq.west);
\draw[-] (fourthcircle.east) -- (firstsq.west);
\draw[-,very thick, dashed] (secondcircle.east) -- (secondsq.west);
\draw[-,very thick, dashed] (thirdcircle.east) -- (fourthsq.west);
\draw[-] (fourthcircle.east) -- (fourthsq.west);
\end{tikzpicture}}
\begin{eqnarray*}
&(i)&\\
a_1& : &p_1\\
a_2& : &p_1,(p_2,p_3)\\
a_33& : &p_1,p_4\\
a_4& : &p_1,p_4
\end{eqnarray*}
\end{minipage}
\begin{minipage}{0.24\textwidth}
%\begin{subfigure}
\scalebox{1.0}{\begin{tikzpicture}[
roundnode/.style={circle, draw=black!100,  fill=black, inner sep=0pt, minimum size=4pt},
squarednode/.style={circle, draw=black!100, fill=black, inner sep=0pt, minimum size=4pt},
]

%Nodes
\node[roundnode]      (maintopic) [label=left:$a_1$]                             {};
\node[roundnode]        (secondcircle)       [below=5mm of maintopic, label=left:$a_2$] {};
\node[roundnode]      (thirdcircle)           [below=5mm of secondcircle, label=left:$a_3$] {};
\node[roundnode]        (fourthcircle)       [below=5mm of thirdcircle, label=left:$a_4$] {};
%Nodes
\node[squarednode]      (firstsq)  [right=of maintopic, label=right:$p_1$]          {};
\node[squarednode]        (secondsq)       [right=of secondcircle, label=right:$p_2$] {};
\node[squarednode]      (thirdsq)       [right=of thirdcircle, label=right:$p_3$] {};
\node[squarednode]        (fourthsq)       [right=of fourthcircle,label=right:$p_4$] {};

%Lines
\draw[-,very thick, dashed] (maintopic.east) -- (firstsq.west);
\draw[-] (thirdcircle.east) -- (firstsq.west);
\draw[-] (thirdcircle.east) -- (thirdsq.west);
\draw[-] (fourthcircle.east) -- (firstsq.west);
\draw[-,very thick, dashed] (secondcircle.east) -- (secondsq.west);
\draw[-,very thick, dashed] (fourthcircle.east) -- (thirdsq.west);
\draw[-] (secondcircle.east) -- (fourthsq.west);
\end{tikzpicture}}
\begin{eqnarray*}
&(ii)&\\
a_1& : &p_1 \\
a_2& : &p_1,(p_2,p_4)\\
a_3& : &p_1,p_3\\
a_4& : &p_1,p_3
\end{eqnarray*}
\end{minipage}
%\begin{subfigure}
\begin{minipage}{0.24\textwidth}
\scalebox{1.0}{\begin{tikzpicture}[
roundnode/.style={circle, draw=black!100,  fill=black, inner sep=0pt, minimum size=4pt},
squarednode/.style={circle, draw=black!100, fill=black, inner sep=0pt, minimum size=4pt},
]

%Nodes
\node[roundnode]      (maintopic) [label=left:$a_1$]                             {};
\node[roundnode]        (secondcircle)       [below=5mm of maintopic, label=left:$a_2$] {};
\node[roundnode]      (thirdcircle)           [below=5mm of secondcircle, label=left:$a_3$] {};
\node[roundnode]        (fourthcircle)       [below=5mm of thirdcircle, label=left:$a_4$] {};
%Nodes
\node[squarednode]      (firstsq)  [right=of maintopic, label=right:$p_1$]          {};
\node[squarednode]        (secondsq)       [right=of secondcircle, label=right:$p_2$] {};
\node[squarednode]      (thirdsq)       [right=of thirdcircle, label=right:$p_3$] {};
\node[squarednode]        (fourthsq)       [right=of fourthcircle,label=right:$p_4$] {};

%Lines
\draw[-,very thick, dashed] (maintopic.east) -- (firstsq.west);
\draw[-] (secondcircle.east) -- (firstsq.west);
%\draw[-] (thirdcircle.east) -- (firstsq.west);
%\draw[-] (fourthcircle.east) -- (firstsq.west);
\draw[-,very thick, dashed] (thirdcircle.east) -- (thirdsq.west);
\draw[-,very thick, dashed] (fourthcircle.east) -- (fourthsq.west);
\draw[-] (fourthcircle.east) -- (secondsq.west);
\end{tikzpicture}}
\begin{eqnarray*}
&(iii)&\\
a_1& : &p_1\\
a_2& : &p_1\\
a_3& : &p_1,p_3\\
a_4& : &p_1,(p_2,p_4)
\end{eqnarray*}
\end{minipage}
\begin{minipage}{0.24\textwidth}
\scalebox{1.0}{\begin{tikzpicture}[
roundnode/.style={circle, draw=black!100,  fill=black, inner sep=0pt, minimum size=4pt},
squarednode/.style={circle, draw=black!100, fill=black, inner sep=0pt, minimum size=4pt},
]

%Nodes
\node[roundnode]      (maintopic) [label=left:$a_1$]                             {};
\node[roundnode]        (secondcircle)       [below=5mm of maintopic, label=left:$a_2$] {};
\node[roundnode]      (thirdcircle)           [below=5mm of secondcircle, label=left:$a_3$] {};
\node[roundnode]        (fourthcircle)       [below=5mm of thirdcircle, label=left:$a_4$] {};
%Nodes
\node[squarednode]      (firstsq)  [right=of maintopic, label=right:$p_1$]          {};
\node[squarednode]        (secondsq)       [right=of secondcircle, label=right:$p_2$] {};
\node[squarednode]      (thirdsq)       [right=of thirdcircle, label=right:$p_3$] {};
\node[squarednode]        (fourthsq)       [right=of fourthcircle,label=right:$p_4$] {};

%Lines
\draw[-,very thick, dashed] (maintopic.east) -- (firstsq.west);
\draw[-] (secondcircle.east) -- (firstsq.west);
\draw[-, very thick, dashed] (secondcircle.east) -- (secondsq.west);
%\draw[-] (fourthcircle.east) -- (firstsq.west);
\draw[-,very thick, dashed] (thirdcircle.east) -- (fourthsq.west);
\draw[-] (thirdcircle.east) -- (thirdsq.west);
\draw[-,very thick, dashed] (fourthcircle.east) -- (thirdsq.west);
\end{tikzpicture}}
\begin{eqnarray*}
&(iv)&\\
a_1& : &p_1\\
a_2& : &p_1,p_2\\
a_3& : &p_1, (p_3,p_4)\\
a_4& : & p_1,p_3
\end{eqnarray*}
\end{minipage}
\caption{Figure indicating possible status changes after deletion of applicant $a_4$: 
Applicants are shown on left whereas posts are shown on right.
Dashed lines indicate a rank-maximal matching $M$ prior to deletion of $a_4$. Preference lists are shown below each figure. In (i), $a_4$ is unmatched. Deletion of $a_4$ keeps $M$ unchanged but status of $a_3$ and $p_4$ changes respectively from even and odd to unreachable. Edge $(a_3,p_1)$ needs to be deleted. 
In (ii), $a_4$ is matched and even. Deletion of $a_4$ 
results in augmenting path and $M$ changes to $(M\setminus\{(a_4,p_3)\})\cup \{(a_3,p_3)\}$. In (iii), $a_4$ is odd. Deletion of $a_4$ does not change the status of any node. In (iv), $a_4$ is unreachable. Deletion of $a_4$ changes the status of $p_3,p_4$ from unreachable to even and that of $a_3$ from unreachable to odd.
Note that some edges incident on $p_1$ have been deleted as they are $\odd\odd$ or $\odd\un$ edges.}
\label{fig:del}
\end{figure}

\begin{appendix-theorem}{\ref{thm:del}}
Algorithm~\ref{algo:del} correctly updates the rank-maximal matching $M$ on deletion of an applicant.
Moreover, it takes time $O(r(m+n))$.
\end{appendix-theorem}
\begin{proof}
If $a$ is unmatched in $M$, clearly $M$ remains unchanged by deletion of $a$. 
As $a$ is assumed to be matched to a post $p$ by a rank $j$ edge, $a$ is even at least for the first $j-1$
stages of Algorithm~\ref{algo:Irving}. Hence the applicants and posts which have alternating paths 
from $a$ are respectively even and odd in all those stages. Deletion of $a$ may make them unreachable, if they do not have alternating paths from another unmatched applicant. This needs
recomputation of labels $\EE,\odd,\un$. Also, higher rank edges on those applicants whose label
changes from $\EE$ to $\un$ need to be deleted. This is done in steps \ref{del:upbeg} to \ref{del:upend} of Algorithm~\ref{algo:del}.

At stage $j$, deletion of $a$ leads to deletion of the edge $(a,p)$ from $M$. 
Hence $p$ becomes free.
The following cases arise:
\begin{description}
\item[Case $1$: $p$ is odd in $G'_j$:]Then there is an alternating path to $p$ from some unmatched applicant in $G'_j$ with respect to $M_j$.
This path now becomes an augmenting path in $H'_j$ with respect to $M'_j$. Checking this case and augmenting along an augmenting path starting from $p$ takes $O(m+n)$ time. Now the posts on this path may not have an alternating
path from an unmatched applicant. In this case, their labels change from $\odd$ to $\unreach$ and
those of the applicants matched to them change from $\EE$ to $\un$. Thus higher rank edges on 
such applicants need to be deleted. This is done in lines \ref{del:oddbeg} to \ref{del:oddend}.

\item[Case $2$: $p$ is unreachable in $G'_j$:] Then there is no alternating path to $p$ with respect 
to $M_j$ from an unmatched 
applicant and hence no augmentation is possible at this stage. In this case, $p$ remains unmatched
in $H'_j$, and hence has the label $\EE$. Labels on applicants and posts which have alternating paths from $p$ change from $\un$ to $\odd$ and $\un$ to $\EE$ respectively. Such posts need to
get back their higher rank edges which were deleted in Algorithm~\ref{algo:Irving}. We include such posts into the set $S$.
This takes $O(m+n)$ time and is done in lines $12$ to $15$.

\item[Case $3$: $p$ is even in $G'_j$: ]Thus $p$ and possibly some more vertices have an alternating path with respect to $M_j$ in
$G'_j$ from an unmatched post $q$. Due to deletion of $(a,p)$ edge, some of these vertices may
no longer have an alternating path from $q$. Labels on such applicants and posts change from 
$\odd$ and $\EE$ respectively to $\un$. The algorithm deletes higher rank edges on such posts.  
\end{description}
Now the algorithm considers subsequent stages.
If $p$ is odd in $G'_j$ above and the matching is augmented as described above, it leads to 
matching an applicant $b$ in $M'_j$ who is unmatched in $M_j$. If, in $M$, $b$ is matched to say
$q$ at a rank $k>j$ then the augmentation will lead to $q$ losing its matched edge at stage $k$. In
this case, same procedure needs to be repeated as above.

Thus the algorithm runs in time $O(m+n)$ for each stage and hence a total of $O(r(m+n))$ time.
As all cases are exhaustively considered above, it updates the matching and the reduced graphs
correctly.
\hfill\qed\smallskip\end{proof}

\end{document}